\documentclass[conference,10pt]{IEEEtran}

\usepackage{amsfonts}
\usepackage{amsmath}
\usepackage{amsthm}
\usepackage{amssymb}
\usepackage{mathtools}
\usepackage{upgreek}
\usepackage{xcolor}
\usepackage{tikz}
\usepackage{tikz-cd}
\usepackage{hyperref}

\definecolor{eric}{HTML}{99c0ff}

\definecolor{matthieu}{HTML}{fff000}

\definecolor{antoine}{HTML}{9fdf9f}

\newtheorem*{thm*}{Theorem}
\newtheorem{thm}{Theorem}
\newtheorem{prop}{Proposition}
\newtheorem{defn}{Definition}
\newtheorem{lem}{Lemma}
\newtheorem{cor}{Corollary}

\newtheorem{thmapdx}{}
\newenvironment{thm-apdx}[1]
  {\renewcommand\thethmapdx{#1}\thmapdx}
  {\endthmapdx}

\newcommand\UU{\mathcal{U}}
\newcommand\SC{\mathcal{S}}
\newcommand\MM{\mathbb{M}}

\newcommand\Idx{\mathsf{Idx}\,}
\newcommand\Cns{\mathsf{Cns}\,}
\newcommand\Typ{\mathsf{Typ}\,}
\newcommand\Pos{\mathsf{Pos}\,}

\newcommand\Id{\mathsf{Id}}
\newcommand\Pb{\mathsf{Pb}\,}
\newcommand\Slice{\mathsf{Slice}\,}

\newcommand{\ooGrp}{\infty\mhyphen\mathsf{Grp}}
\newcommand{\preooCat}{\mathsf{pre}\mhyphen\infty\mhyphen\mathsf{Cat}}

\newcommand{\dsum}[1]{\textstyle{\sum_{(#1)}}\,}
\newcommand{\dprod}[1]{\textstyle{\prod_{(#1)}}\,}
\mathchardef\mhyphen="2D

\newcommand\refl{\mathsf{refl}}
\newcommand\ttt{\mathsf{tt}}
\newcommand\ctr{\mathsf{ctr}\,}
\newcommand\wit{\mhyphen\mathsf{wit}}
\newcommand\coh{\mhyphen\mathsf{coh}}
\newcommand\alg{\mhyphen\mathsf{alg}}

\newcommand{\da}{{\downarrow}}

\newcommand\MMd{\mathbb{M\da}}
\newcommand\Idxd{\mathsf{Idx}\da\,}
\newcommand\Cnsd{\mathsf{Cns}\da\,}
\newcommand\Typd{\mathsf{Typ}\da\,}
\newcommand\Posd{\mathsf{Pos}\da\,}

\newcommand\upetad{\upeta\da}
\newcommand\upmud{\upmu\da}

\newcommand\lfd{\operatorname{lf\da}}
\newcommand\ndd{\operatorname{nd\da}}

\newcommand\Idd{\mathsf{Id\da}\,}
\newcommand\Pbd{\mathsf{Pb\da}\,}
\newcommand\Sliced{\mathsf{Slice\da}\,}

\newcommand\Unit{\top}
\newcommand\Empty{\bot}
\newcommand\botelim{\bot\mhyphen\mathsf{elim}}
\newcommand\etapos{\upeta\mhyphen\mathsf{pos}\,}
\newcommand\etaposelim{\upeta\mhyphen\mathsf{pos}\mhyphen\mathsf{elim}\,}
\newcommand\etadec{\upeta\mhyphen\mathsf{dec}\,}
\newcommand\etadecd{\upeta\mhyphen\mathsf{dec}\da\,}
\newcommand\mupos{\upmu\mhyphen\mathsf{pos}\,}
\newcommand\muposfst{\upmu\mhyphen\mathsf{pos}\mhyphen\mathsf{fst}\,}
\newcommand\mupossnd{\upmu\mhyphen\mathsf{pos}\mhyphen\mathsf{snd}\,}

\newcommand\lf{\mathsf{lf}\,}
\newcommand\nd{\mathsf{nd}\,}
\newcommand\fst{\mathsf{fst}\,}
\newcommand\snd{\mathsf{snd}\,}
\newcommand\inl{\mathsf{inl}\,}
\newcommand\inr{\mathsf{inr}\,}

\newcommand\iscontr{\mathsf{is}\mhyphen\mathsf{contr}\,}
\newcommand\ismult{\mathsf{is}\mhyphen\mathsf{mult}\,}
\newcommand\isfibrant{\mathsf{is}\mhyphen\mathsf{fibrant}\,}
\newcommand\carmult{\mathsf{car}\mhyphen\mathsf{is}\mhyphen\mathsf{mult}\,}
\newcommand\relfib{\mathsf{rel}\mhyphen\mathsf{is}\mhyphen\mathsf{fibrant}\,}
\newcommand\isalgebraic{\mathsf{is}\mhyphen\mathsf{algebraic}\,}

\newcommand\OpType{\mathsf{OpetopicType}\,}
\newcommand\OvrOpType{\da\mathsf{OpType}\,}
\newcommand\Car{\mathcal{C}\,}
\newcommand\Rel{\mathcal{R}\,}

\newcommand{\commentt}[1]{}

\begin{document}

\setlength{\abovedisplayskip}{6pt}
\setlength{\belowdisplayskip}{6pt}

\title{Types Are Internal $\infty$-Groupoids}
\author{
  \IEEEauthorblockN{Eric Finster}
  \IEEEauthorblockA{Cambridge University\\
    Department of Computer Science \\
    ericfinster@gmail.com
  }
  \and
  \IEEEauthorblockN{Antoine Allioux}
  \IEEEauthorblockA{Inria \&
   IRIF, Université de Paris\\
   France\\
   antoine.allioux@irif.fr}
  \and
  \IEEEauthorblockN{Matthieu Sozeau}
  \IEEEauthorblockA{Inria \&
    LS2N, Université de Nantes\\
    France\\
    matthieu.sozeau@inria.fr}
}

\IEEEoverridecommandlockouts
\IEEEpubid{\makebox[\columnwidth]{978-1-6654-4895-6/21/\$31.00~
\copyright2021 IEEE \hfill} \hspace{\columnsep}\makebox[\columnwidth]{ }}
\maketitle

\begin{abstract}
  By extending type theory with a universe of definitionally
  associative and unital polynomial monads, we show how to arrive at a
  definition of \emph{opetopic type} which is able to encode a number
  of fully coherent algebraic structures.  In particular, our approach
  leads to a definition of $\infty$-groupoid internal to type theory
  and we prove that the type of such $\infty$-groupoids is equivalent
  to the universe of types.  That is, every type admits the structure
  of an $\infty$-groupoid internally, and this structure is unique.
\end{abstract}

%
\IEEEpeerreviewmaketitle

\section{Introduction}

Homotopy Type Theory has brought a new perspective to intensional
Martin-L\"{o}f type theory: the higher identity types of a type endow
it with the structure of an \emph{$\infty$-groupoid}, and ideas from
homotopy theory provide us with a means to predict and understand the
resulting tower of identifications.  While this perspective has been
enormously clarifying with respect to our understanding of the notion
of proof-relevant equality, leading, as it has, to a new class of
models as well as new computational principles, a number of
difficulties remain in order to complete the vision of type theory as
a foundation for a new, structural mathematics based on
homotopy-theoretic and higher categorical principles.

Foremost among these difficulties is the following: how does one
describe a well behaved theory of \emph{algebraic structures} on
arbitrary types?  The fundamental difficulty in setting up such a
theory is that, in a proof relevant setting, nearly all of the
familiar algebraic structures (monoids, groups, rings, and categories,
to take a few) become \emph{infinitary} in their presentation.
Indeed, the axioms of these theories, which take the form of a finite
list of mere \emph{properties} when the underlying types are sets,
constitute additional \emph{structure} when they are no longer assumed
to be so.  Consequently, in order to arrive at a well-behaved theory, the
axioms themselves must be subject to additional axioms, frequently
referred to generically as ``coherence conditions''. In short, in a
proof relevant setting, it no longer suffices to describe the
equations of an algebraic structure at the ``first level'' of
equality.  Rather, we must specify how the structure behaves
\emph{throughout the entire tower} of identity types, and this is an
infinite amount of data.  How do we organize and manipulate this data?

Similar problems have arisen in the mathematics of homotopy theory and
higher category theory, and many solutions and techniques are known.
Bafflingly, however, all attempts to import these ideas into plain
homotopy type theory have, so far, failed.  This appears to be a
result of a kind of circularity: all of the known classical techniques
at \emph{some} point rely on set-level algebraic structures themselves
(presheaves, operads, or something similar) as a means of presenting
or encoding higher structures.  Internally to type theory, however, we
do not have recourse to such techniques.  Indeed, without further
hypotheses, we do not even expect that the most basic objects of the
theory, types themselves, are presented by set-level structures. This
leaves us in a strange position: absent a theory of algebraic
structures, we have nothing to use to encode algebraic structures!

We suggest that a possible explanation for this phenomenon is the
following: contrary to our experience with set-level mathematics,
where an algebraic structure (i.e. a ``structured set'') can itself be
defined just in terms of sets: underlying sets, functions, sets of
relations and so on, when we pass to the world of homotopy theoretic
mathematics, the notion of \emph{type} and \emph{structured type} are
simply no longer independent of each other in the same way.
Consequently, some primitive notion of structured type must be defined
\emph{at the same time} as the notion of type itself.  The present
work is a first attempt at rendering this admittedly somewhat vague
idea precise.

We begin by imagining a type theory which, in addition to defining a
universe $\UU$ of \emph{types}, defines at the same time a universe
$\SC$ of \emph{structures}. Of course, we will need to be somewhat
more precise about what exactly we mean by \emph{structure}.  Category
theory suggests that one way of representing a structure is by the
monad on $\UU$ which it defines, so we might think of $\SC$ as a
universe of monads.  In practice, however, it will be useful to
restrict to a particularly well behaved class of monads, having
reasonable closure properties, and for which we have a good
understanding of their higher dimensional counterparts.  We submit
that a reasonable candidate for such a well-behaved collection is the
class of \emph{polynomial monads}~\cite{GK}.

We feel that this is an appropriate class of structures for a number
of reasons.  A first reason is that this class of monads arises quite
naturally in type theory already: indeed, a large literature exists on
the interpretation of inductive and coinductive types as initial
algebras and terminal coalgebras for polynomial monads, and we
consider our work as deepening and extending this connection.
Furthermore, this class of algebraic structures enjoys some pleasant
properties which make them particularly amenable to ``weakening''.
For example, the very general approach to weakening algebraic
structures developed by Baez and Dolan in \cite{BD98} can be smoothly
adapted to the polynomial case. While the cited work employs the
language of symmetric operads, connections with the theory of
polynomial functors were already described in \cite{KJBM}, and
moreover, recent work \cite{gepner2017infty} has shown that, in type
theory, we should expect symmetric operads to in fact be
\emph{subsumed} by the theory of polynomial monads.

The central intuition of Baez and Dolan's approach, is that each
polynomial monad $M$ determines its \emph{own} higher dimensional
collection of shapes (the $M$-opetopes) generated directly from the
syntactic structure of its terms.  They go on to introduce the notion
of an \emph{$M$-opetopic type} which is, roughly, a collection of well
formed decorations of these shapes, and the notion of weak $M$-algebra
is then defined as an $M$-opetopic type satisfying certain closure
properties.  In this sense, their approach differs from, say,
approaches based on simplices, cubes or spheres in that the geometry
is not fixed ahead of time, but adapted to the particular structure
under consideration.

With these considerations in mind, our plan in the present work is to
put the idea of a type theory with primitive structures to the test.
What might it look like, and what might it be able to prove?  In order
to answer these questions, we will build a prototype of the theory
\footnote{The Agda source is available here:
  \url{https://github.com/ericfinster/opetopic-types}} in the proof
assistant Agda, and exploit the recent addition of \emph{rewrite
  rules} \cite{cockx:hal-02901011} which permits us to extend
definitional equality by new well-typed reductions.  The use of such
rewrites is necessary to ensure that our primitive structures are not
subject to the same infinite regress of coherence conditions which has
so far obstructed more naive attempts and defining such objects.

Concretely, we will introduce a universe $\MM$,\footnote{We mean by
  this notation to distinguish the universe $\MM$ of this
  \emph{particular} implementation from the generic idea of a universe
  of structures $\mathcal{S}$, the properties of which we expect to be
  refined by further investigation.} whose elements we think of as
codes for polynomial monads and describe the structures they decode
to.  Because we think of the objects of the universe $\MM$ as
primitives of our theory, on the same level as types, we allow
ourselves the freedom to prescribe their computational behavior: in
particular, we will equip them with definitional associativity and
unit laws using the rewrite mechanism alluded to above.  We emphasize
that if structures are taken as defined in parallel with types, then
this kind of definitional behavior should be no more surprising than,
say, the definitional associativity of function composition.

We then show how the existence of our universe $\MM$ has some strong
consequences.  In particular, it allows us to implement the Baez and
Dolan definition of opetopic type alluded to above, and subsequently
to define a number of weak higher dimensional structures.  Among the
structures which we are able to define using this technique are
$\mathbb{A}_\infty$-monoids and groups, $(\infty,1)$-categories and
presheaves over them (in particular, our setup leads to a definition
of \emph{simplicial type}), and as a special case, $\infty$-groupoids
themselves.

There arises, then, the problem of justifying the correctness of our
definitions.  In order to do so, we will take up the example of
$\infty$-groupoids in some detail.  Indeed, since the homotopical
interpretation of type theory asserts that types should ``be''
$\infty$-groupoids, it seems natural to compare these two objects.
Our main result is the following:
\begin{thm*}
  There is an equivalence
  \[ \UU \simeq \infty\mhyphen\mathsf{Grp} \]
\end{thm*}
In other words, every type admits the structure of an
$\infty$-groupoid in our sense, and that structure is unique.  This
theorem, therefore, can be regarded as a (constructive)
internalization of the intuition provided by the various
meta-theoretic results to this effect~\cite{van2011types,
  lumsdaine2009weak}.

\subsection{Related Work}
\label{sec:related-work}

The so-called \emph{coherence problem}, which is the main motivation
for the present work, has been considered by a number of authors.  We
briefly compare our approach with two prominent other strains of
thought.

\subsubsection{Synthetic Structures}

One way to avoid some of the problems posed by the definition of
higher dimensional structures is to simply enlarge the collection of
basic objects to include them.  Such approaches may be described as
synthetic, in that they do not reduce higher structures to more
primitive objects in exactly the same way that homotopy type theory
itself does not define $\infty$-groupoids in terms of sets.  This
point of view is often adopted, for example, in the research into
\emph{directed type theories}, whether they be aimed at specific
structures like higher categories as in \cite{north2019towards}, or
allow for more general directed spaces as in \cite{riehl2017type}.

While our theory does indeed add some new primitive structures to type
theory, we collect these structures in a universe and decode them into
collections of ordinary types and maps between them.  Moreover, we go
on to use these additional strict structures to give analytic,
internal definitions of higher structures.

\subsubsection{Two-Level Type Theory}

Perhaps the closest related work to the current approach is that of
Two-Level Type Theory \cite{DBLP:journals/corr/AnnenkovCK17}.  There
it is advocated to add a second ``level'' to type theory with a
set-truncated equality type which one can then use to make
meta-theoretic statements about the inner level, whose objects are
typically taken to be the homotopically meaningful ones.  The
two-level approach provides, thus, a great deal of generality and
flexibility at the cost of restricting the applicability of
homotopical interpretation of types to the inner theory.  It is
likely, for example, that our theory could be developed inside a
two-level system and many of the rewrites we employ proven as theorems
in the outer level.  By constrast, our approach is, we feel, somewhat
more economical, extending the theory with a specific set of rewrite
rules, and pointing towards the possibility of a useful theory of
higher structures without the need to restrict homotopical principles
like univalence.

\subsection{Preliminaries}

The basis of our metatheory is the type theory implemented in
Agda~\cite{agda} which is an extension of the predicative part of
Martin-Löf type theory~\cite{ML75}. Among the particular types that
Agda implements, we shall use inductive types, records and coinductive
records.

As such, we adopt a style similar to Agda code, writing
$(x : A) \to B\, x$ for the dependent product (although we
occasionally also employ the $\prod_{(x : A)} B\, x$ notation if it
improves readability).  We also make use of the implicit counterpart
of the dependent products, written $\{x : A\} \to B\, x$.  This allows
us to hide arguments which can be inferred from the context and hence
clarify our notation.  Non-dependent functions are denoted $A \to B$,
as usual. Functions enjoy the usual $\eta$ conversion rule.

We shall make extensive use of coinductive record types, as well as
copatterns for producing elements of these types.  We write $\Unit$
for the empty record with a constructor $\ttt : \Unit$. We write
$\dsum{x : A} B\ x$ for the dependent sum as a record with constructor
$\_, \_$ and projections $\fst$ and $\snd$. Pairs which are not
dependent are denoted $A \times B$.

We write $\bot$ for the empty type, using absurd patterns where
appropriate, and writing $\botelim$ for the unique function for $\bot$
to any type.

The identity type $\_\equiv\_ : \{A : \UU\}\ (x\ y : A) \to \UU$ is
an inductive type with one constructor
$\refl : (x : A) \to x \equiv x$.

We shall make use of the notion of contractible type denoted
$\iscontr$ whose center of contraction will be referred to as
$\ctr$. Other notions defined in the HoTT book~\cite{hottbook} will be
employed including equivalence of types denoted $X \simeq Y$, function
extensionality denoted $\mathsf{funext}$, as well as the Univalence
Axiom.

We write $\UU$ for the universe of small types, and $\UU_1$ for the
next universe when necessary.

In order to unclutter and clarify the presentation, we
occasionally take liberties with the formal definitions, for example,
silently inserting applications to functional extensionality when
necessary, or reassociating $\Sigma$-types in order to avoid a
proliferation of commas. Our formal development in Agda allows no
such informalities to remain.


\section{A Universe of Polynomial Monads}

As we have explained in the introduction, type theory appears to lack
the ability to speak about infinitely coherent algebraic structures,
and our strategy for addressing this problem will be to distinguish a
collection of such structures which we consider as defined by the
theory itself.  We do so using a common technique in the type theory
literature: that of introducing a \emph{universe}.  We write
$\MM : \UU$ for our universe, and we think of its elements as
\emph{codes for polynomial monads}.  Just as a typical type theoretic
universe has some collection of base types and some collection of type
constructors, so we will populate our universe with a collection of
``base monads'' and ``monad constructors''.  In other words: we
will have a syntax of structures which parallels the syntax for types.

Typically, a universe of types $\mathbb{U}$ comes equipped with a
decoding function $El : \mathbb{U} \to \UU$.  In the case of our
universe of monads $\MM$, each element $M : \MM$ will decode not to a
single type, but to a collection of types and type families equipped
with some structure.  We will use rewrite rules to specify
the computational behavior of this structure.

\subsection{Polynomial Structure}

To begin, we first equip each $M : \MM$ with an underlying
\emph{polynomial} or \emph{indexed container} \cite{DBLP:journals/jfp/AltenkirchGHMM15}.
This is accomplished by postulating the following collection of decoding functions:
\begin{align*}
  \Idx &: \MM \to \UU\\
  \Cns &: (M : \MM) \to \Idx M \to \UU\\
  \Pos &: (M : \MM)\ \{i : \Idx M\} \to \Cns M\, i \to \UU \\
  \Typ &: (M : \MM)\ \{i : \Idx M\}\ (c : \Cns M i) \\
       &\to \Pos M\, c \to \Idx M 
\end{align*}
Polynomials of this sort appear in the computer science literature as
the ``data of a datatype declaration''.  They can equivalently be seen
as a way to describe the signature of an algebraic theory: concretely,
the elements of $\Idx M$, which we refer to as \emph{indices} serve as
the sorts of the theory, and for $i : \Idx M$, the type $\Cns M\, i$
is the collection of operation symbols whose ``output'' sort is $i$.
The type $\Pos M\, c$ then assigns to each operation a collection of
``input positions'' which are themselves assigned an index via the
function $\Typ$.

It follows that every monad $M$ induces a functor
$[\_] : (\Idx M \to \UU) \to (\Idx M \to \UU)$ called its
\emph{extension} given by
\begin{align*}
  [ M ]\, X\, i &= \sum_{(c : \Cns M\, i)} (p : \Pos M\, c) \to X \, (\Typ M\, c\ p)
\end{align*}
We may think of the value of this functor at a type family
$X : \Idx M \to \UU$ as the type of \emph{constructors of $M$ with
  inputs decorated by elements of $X$}.  Indeed, we will frequently
refer to a dependent function of the form
\[ (p : \Pos M\, c) \to X \, (\Typ M\, c\ p) \] where $X$ is as above,
as a \emph{decoration} of $c$ by elements of $X$.

\subsection{Monadic Structure}
\label{sec:mnd-struct}

Next, for each monad $M : \MM$, we are going to equip the
underlying polynomial of $M$ with an algebraic structure:
specifically, that structure required on the underlying
polynomial so that the associated extension $[ M ]$ becomes
a \emph{monad}.  In the case at hand, this takes the form
of a pair of functions
\begin{align*}
  \upeta &: (M : \MM)\ (i : \Idx M) \to \Cns M\, i\\
  \upmu &: (M : \MM)\ \{i : \Idx M\}\ (c : \Cns M\, i) \\
         &\to (\delta : (p : \Pos M\, c) \to \Cns M\, (\Typ M\, c\ p)) \\
         &\to \Cns M\, i
\end{align*}
which equip $M$ with a multiplication and unit operation.  We remark
that the second argument $\delta$ of the multiplication $\upmu$ is a
decoration of $c$ in the family $\Cns M$ of constructors, so that we
can think of the input to this function as a ``two-level'' tree.

Crucial for what follows will be that the monads we consider are
\emph{cartesian} in the sense of \cite{GK}.  Type theoretically, the
means we require each monad $M$ to come equipped with equivalences
\begin{align*}
  \Pos M\, (\upeta\ M\, i) &\simeq \Unit \\
  \Pos M\, (\upmu\ M\, c\, \delta) &\simeq \dsum{p : \Pos M\, c} \Pos M\, (\delta\, p)
\end{align*}
Since we are already modifying the definitional equality of our type
theory, it may be tempting to require these equivalences
definitionally by asserting that the type of positions reduces when
applied to constructors of the appropriate form. However, this will
not work: as we will see below, when we come to populate our universe
with concrete monads and monad constructors, the equivalences we find
are often in fact not definitional, even if they remain provable.  As
an alternative, we will equip each monad with introduction,
elimination and computation rules for its positions which will in
effect guarantee that we always have the required equivalence.  Each
monad definition will then be required to implement these rules in a
manner consistent with the various required typing laws.

In the case of $\upeta$, for example, we postulate introduction and
elimination rules of the form
\begin{align*}
  \etapos &: (M : \MM)\ (i : \Idx M) \to \Pos M\, (\upeta\ M\, i)\\
  \etaposelim &: (M : \MM)\ (i : \Idx M) \\
              &\to (X : (p : \Pos M\, (\eta\ M\, i)) \to \UU)\\ 
              &\to (u : X\ (\etapos M\, i)) \\
          &\to (p : \Pos M\, (\upeta\ M\, i)) \to X\ p
\end{align*}
with typing rule
\begin{equation}
  \label{rewrite:etaTyp}
  \tag{Typ-$\eta$}
  \begin{aligned}
    \Typ M\, (\upeta\ M\, i)\ p & \leadsto i 
  \end{aligned}
\end{equation}
and computation rule
\begin{align*}
  \etaposelim\ M\, i\ X\ u\ (\etapos M\, i) &\leadsto u \\
\end{align*}
Notice these are exactly the rules for an inductively defined indexed unit
type.\footnote{In principle, we would also like to have an
  $\eta$-rule for the unit $\upeta$, (that is, we would prefer the
  negative version as we have below for $\upmu$) but unfortunately
  this is not possible with the current implementation of rewriting in
  Agda.}  In particular, decorations of the constructor $\upeta\, M\, i$
in a type family $X : \Pos M\, (\upeta\, M\, i) \to \UU$ are completely determined
by a single element $x : X\, i$, a fact which we record in the following
definition to reduce clutter below:
\begin{align*}
  \etadec M\, X\, \{i\}\, x = \etaposelim M\, i\, (\lambda\, \_ \to X i)\, x 
\end{align*}

Next, for the multiplication $\upmu$, our rules simply mimic the
pairing and projections of the dependent sum.  That is, we
postulate an introduction rule
\begin{align*}
  \mupos &: (M : \MM)\ \{i : \Idx M\}\ \{c : \Cns M\, i\}\ \\
          &\to \{\delta : (p : \Pos M\, c) \to \Cns M\, (\Typ M\, c\ p)\}\\ 
          &\to (p : \Pos M\, c) \to (q : \Pos M\, (\delta\ p))\\
         &\to \Pos M\, (\upmu\ M\, c\ \delta)
\end{align*}
and elimination rules
\begin{align*}
  \muposfst &: (M : \MM)\ \{i : \Idx M\}\ \{c : \Cns M\, i\}\\
            &\to \{\delta : (p : \Pos M\, c) \to \Cns M\, (\Typ M\, c\ p)\}\\
            &\to \Pos M\ (\upmu\ M\, c\ \delta) \to \Pos M\, c\\
  \mupossnd &: (M : \MM)\ \{i : \Idx M\}\ \{c : \Cns M\, i\}\\
            &\to \{\delta : (p : \Pos M\, c) \to \Cns M\, (\Typ M\, c\ p)\}\\
            &\to (p : \Pos M\, (\upmu\ M\, c\ \delta))\\
            &\to \Pos M\, (\delta\ (\muposfst M\, p))
\end{align*}
with typing rule
\begin{align*}
  \label{rewrite:muTyp}
  \tag{Typ-$\mu$}
  \Typ M\, &(\upmu\ M\, c\ \delta)\ p \leadsto \\
           &\Typ M\, (\delta\ (\muposfst M\, p))\ \\
           &\hspace{1.5cm} (\mupossnd M\, p)
\end{align*}
and computation rules
\begin{align*}
  \label{rewrite:muposfst}
  \tag{$\mu$-pos-fst}
  \muposfst M\, (\mupos M\, p\ q) &\leadsto p\\
  \label{rewrite:mupossnd}
  \tag{$\mu$-pos-snd}
  \mupossnd M\, (\mupos M\, p\ q) &\leadsto q\\
  \label{rewrite:muposeta}
  \tag{$\mu$-pos-$\eta$}
  \mupos M\, (\muposfst M\, p)\, (\mupossnd M\, p) &\leadsto p \\
\end{align*}
With the handling of positions in place, we can now state the
unitality and associativity axioms for the monads in our universe.
These take the form of reductions:
\begin{align*}
  \tag{$\mu$-$\eta$-r}
  \label{rewrite:muetar}
  &\upmu\ M\, c\ (\lambda\ p \to \upeta\ M\, (\Typ M\, c\ p)) \leadsto c\\
  \tag{$\mu$-$\eta$-l}\label{rewrite:muetal}
  &\upmu\ M\, (\upeta\ M\, i)\ \delta \leadsto \delta\ (\etapos M\, i)\\
  \tag{$\mu$-$\mu$}\label{rewrite:mumu}
  &\upmu\ M\, (\upmu\ M\, c\ \delta)\ \epsilon \leadsto \\
  &\quad \upmu\ M\ c\ (\lambda\ p \to\ \upmu\ M\, (\delta\ p)\ (\lambda\ q \to \epsilon\ (\mupos M\, p\ q)))
\end{align*}
Additionally, we must posit laws which assert that the constructors
and eliminators for positions are compatible with these equations.  We
omit these for brevity, but the interested reader may consult the
development for details.

While we will not undertake an extensive investigation of the
meta-theoretic properties of our system in this article, we wish to
pause briefly to make at least of few observations to justify its
well-formedness.  For example, there are critical pairs in the rewrite
equations for the monad laws (between the first equation and the
others) so we need to ensure confluence and termination.
\begin{lem}[Strong confluence for $\eta$ and $\mu$]
  The rewrite rules are strongly confluent \cite{huet80}, hence globally confluent.
\end{lem}
\begin{proof}
  The rewrite system is strongly confluent using the rules 
  \eqref{rewrite:etaTyp} and \eqref{rewrite:muTyp} and the associated reduction rules
  for $\mupos$.
  We show the case for \ref{rewrite:muetar} and \ref{rewrite:muetal}.
  We omit $M$ which is fixed here.

  \begin{align*}
  \upmu\ (\upeta\ i)&\ (\lambda\ p \to \upeta\ (\Typ (\upeta\ i)\ p)) \leadsto_{\text{\ref{rewrite:muetar}}} \upeta\ i\\
  \upmu\ (\upeta\ i)&\ (\lambda\ p \to \upeta\ (\Typ (\upeta\ i)\ p)) \leadsto_{\text{\ref{rewrite:muetal}}} \\
  &(\lambda\ p \to \upeta\ (\Typ (\upeta\ i)\ p))\ (\etapos i) \leadsto_\beta \\ 
  &\upeta\ (\Typ (\upeta\ i)\ (\etapos i)) \leadsto_{\text{\ref{rewrite:etaTyp}}} \upeta\ i
  \end{align*}

  The resolution of the \ref{rewrite:muetar}/\ref{rewrite:mumu} pair
  can be found in the appendix of the extended version of this article
  \cite[Lemma \ref{proof:muetarmumu-proof}]{allioux2021types}.
\end{proof}

\vspace{1ex}

\begin{prop}[Termination of rewriting]
  All of the above rules form a terminating rewrite system.
\end{prop}
\begin{proof}
  The $\mupos$, $\etapos$ and \eqref{rewrite:etaTyp} rewrite rules are obviously terminating.
  For \eqref{rewrite:muTyp},  \eqref{rewrite:muetar},
  \eqref{rewrite:muetal} and \eqref{rewrite:mumu}, we need to 
  use a dependency-pairs path ordering as introduced by \cite{blanqui19}
  to verify termination. In particular for associativity, a 
  lexicographic lifting of the subterm relation is not 
  enough to verify \eqref{rewrite:mumu}'s termination as we are 
  going under binders and applying the $\delta$ and $\epsilon$ 
  functions to subterms. This is a variant of the ordinal type 
  eliminator proven to terminate in \cite[Example 14, p11]{blanqui19},
  which requires to ensure that the constructor types of our monads 
  are inductively generated. All the monads considered in this article satisfy this.
\end{proof}

The instances of the $\mu$ and $\eta$ operations for specific monads
will themselves be defined by structural recursion on inductive
datatypes and can be shown to respect the associativity and unitality
laws prositionally. Results such as can be found in~\cite[Lemma
6.8]{cockx:hal-02901011}, therefore, guarantee the consistency of the
system.  Furthermore, we conjecture that the rewrite system is
strongly normalizing in conjunction with the definitional equality of
Agda.

\subsection{Populating the Universe}

In the previous section, we described the generic structure
associated to every monad $M : \MM$.  We now proceed to implement this
structure in concrete cases, describing in each case the most salient
features and omitting unnecessary details where we feel it will improve
the presentation.  Complete definitions can be found in the Agda
formalization.

In the material which follows, we allow ourselves the freedom to use
standard techniques such as inductive definitions and pattern matching
during the definition of each monad.  In practice, this agrees with
the implementation: there, we first define all the necessary structure
using ordinary Agda definitions and subsequently install rewrites
which connect the decoding functions to their desired implementations.
So for example, in order to define the indices of the identity monad
(see below), we first make an ordinary Agda definition
\begin{align*}
  \mathsf{IdIdx} &: \UU \\
  \mathsf{IdIdx} &= \Unit
\end{align*}
and then postulate the rewrite
\[ \Idx \Id \leadsto \mathsf{IdIdx} \]
In the presentation which follows, we omit this auxiliary step
and just write ``$=$'' when defining the structure associated
to each monad.

\subsubsection{The identity monad}

We begin by adding a constant $\Id : \MM$ to the universe to represent
the \emph{identity monad} (so named since its extension induces the
identity monad on $\UU$ up to equivalence).  The polynomial part of
$\Id$ decodes as follows:
\begin{align*}
  &\Idx \Id\, &= \Unit\\
  &\Cns \Id\, \ttt &= \Unit\\
  &\Pos \Id\, \ttt &= \Unit\\
  &\Typ \Id\, \ttt\ \ttt &= \ttt
\end{align*}
Given the triviality of the associated polynomial, it is perhaps not
surprising that its unit and multiplication are equally trivial.  Indeed,
they are given by:
\begin{align*}
  &\upeta\, \Id\, i = \ttt \\
  &\upmu\, \Id\, \_\, \delta = \delta\, \ttt
\end{align*}
We omit the remaining structure, which has a similar flavor.

\subsubsection{The pullback monad}

Our first monad constructor starts from a monad $M : \MM$ and a family
$X : \Idx M \to \UU$ and refines the indices of $M$ by additionally
decorating the inputs and output of each constructor by
elements of $X$.  We refer to the resulting monad as the
\emph{pullback of $M$ along $X$} (cf. \cite[Section 2.4]{BD98}).  We
implement this construction by first postulating a function
\[ \Pb : (M : \MM)\, (X : \Idx M \to \UU) \to \MM \] which adds the
necessary code to our universe.  We next define the polynomial part of
$\Pb M\, X$ as follows:
\begin{align*}
  &\Idx (\Pb M\, X) &&= \dsum{i : \Idx M} X\ i\\
  &\Cns (\Pb M\, X)\ (i , x) &&= \\
  &&& \hspace{-2cm}\dsum{c : \Cns M\, i} \dprod{p : \Pos M\, c} X\, (\Typ M\, c\, p) \\
  &\Pos (\Pb M\, X)\ (c , \nu) &&= \Pos M\, c\\
  &\Typ (\Pb M\, X)\ (c , \nu)\ p &&= (\Typ M\, c\ p\, ,\, \nu\ p)
\end{align*}

The unit for the pullback monad simply calls the unit of the
underlying monad and decorates its input with the same value as its
output:
\[ \upeta\ (\Pb M\, X)\ (i \,, x) = (\upeta\ M\, i\, ,\, \etadec M\, X\, x) \]

As for the multiplication of the pullback monad, it again simply calls
the multiplication of the underlying monad, this time decorating the
result using the decorations of the second-level constructors,
and forgetting the intermediate decoration.  That is, we have
\[ \upmu\ (\Pb M\, X)\ (c \,, \nu)\ \delta = (\upmu\ M\, c\ \delta', \nu') \]
where
\begin{align*}
  \delta'\ p &= \fst\, (\delta\, p)\\
  \nu'\ p &= \snd\, (\delta\, (\muposfst p))\, (\delta\, (\mupossnd p))
\end{align*}
The remaining structure is easily worked out from these definitions,
and we omit the details.

\subsubsection{The Slice Monad}
\label{sec:slice-monad}

The Baez-Dolan slice construction is at the heart of the opetopic
approach: it is this construction which allows us to ``raise the
dimension'' of the coherences in our algebraic structures.  In our
setting, it will take the form of a monad constructor
$\Slice : \MM \to \MM$.  The basic intuition is that, for a monad
$M : \MM$, the monad $\Slice M$ may be described as the \emph{monad of
  relations in $M$}.  In order to realize this intuition, we have to
find a way to encode the relations of $M$ as some kind of data, just
as the identity type encodes the relations in an ordinary type as
data.  This data will then serve as the constructors for the slice
monad.

To begin, for a monad $M : \MM$, let us define
\[ \Idx (\Slice M) = \dsum{i : \Idx M} \Cns M\, i \] That is, the
indices of the monad $\Slice M$ are exactly the constructors of the
monad $M$.  Next, we are going to capture the notion of \emph{relation
  in $M$} with the help of a certain inductive family, defined as
follows:
\begin{align*}
  &\texttt{data}\hspace{1ex} \mathsf{Tree} : \Idx (\Slice M) \to \UU \hspace{1ex} \texttt{where} \\
  &\hspace{.2cm}\lf : (i : \Idx M) \to \mathsf{Tree}\ (i\, , \upeta\ M\, i) \\
  &\hspace{.2cm}\nd : \{i : \Idx M\}\ (c : \Cns M\, i)\\
  &\hspace{.5cm}\to (\delta : (p : \Pos M\, c) \to \Cns M\, (\Typ M\, c\ p))\\ 
  &\hspace{.5cm}\to (\epsilon : (p : \Pos M\, c) \to \mathsf{Tree}\ (\Typ M\, c\ p, \delta\ p))\\ 
  &\hspace{.5cm}\to \mathsf{Tree}\ (i\, , \upmu\ M\, c\ \delta)
\end{align*}
And we define $\Cns (\Slice M) = \mathsf{Tree}$.

The reader familiar with the theory of inductive types may recognize
this as a modified form of the \emph{indexed $W$-type} associated to a
polynomial or indexed container.  Here, as in that case, the elements
of this type are \emph{trees} generated by the constructors of the
polynomial in question (the underlying polynomial of $M$, in the case
at hand).  The difference in the present setup is that our polynomial
is equipped with a multiplication and unit, and we reflect this fact
by indexing our trees not just over the indices (as is typically the
case) but also over the constructors, applying the multiplication and
unit as appropriate.  The result is that we may view an element
$\sigma : \Cns (\Slice M)\, (i \,, c)$ as ``a tree generated by the
constructors of $M$ whose image under iterated multiplication is
$c$''.  It is in this sense that this definition captures the
\emph{relations} in the original monad $M$.

We now turn to the rest of the structure required to complete the
definition of $\Slice M$.  Intuitively speaking, the positions of
a tree $\sigma$ will be its \emph{internal nodes}.  This can be
accomplished by defining the positions by recursion on the
constructors as follows:
\begin{align*}
  &\Pos (\Slice M)\, (\lf\, i) &&= \Empty \\
  &\Pos (\Slice M)\, (\nd\, c\, \delta\, \epsilon)
                               &&= \\
  &&&\hspace{-1cm} \Unit \sqcup \sum_{(p : \Pos M\, c)} \Pos (\Slice M)\, (\epsilon\, p)
\end{align*}
In other words, if our tree is a leaf, it has no positions, and if it
is a node, its positions consist of either the unit type (to record
the current node) or else the choice of a position of the base
constructor and, recursively, a node of the tree attached to that
position.

Finally, the typing function $\Typ (\Slice M)\, \sigma\, p$ just
projects out the constructor of $M$ occurring at the node of $\sigma$
specified by position $p$:
\begin{align*}
  &\Typ (\Slice M)\, (\lf\, i)\, () \\
  &\Typ (\Slice M)\, (\nd\, \{i\}\, c\, \delta\, \epsilon)\, (\inl \ttt) &&= (i \,, c) \\
  &\Typ (\Slice M)\, (\nd\, \{i\}\, c\, \delta\, \epsilon)\, (\inl (p \,, q)) &&= \\
  &&&\hspace{-2cm} \Typ (\Slice M)\, (\epsilon\, p)\, q
\end{align*}

It remains to describe the unit and multiplication of the slice monad.
In accordance with the general laws for monads, the unit constructor
needs to have a unique position, and since the positions of a given
tree are given by occurrences of constructors, this implies that
the unit at a given constructor $c$ should be the \emph{corolla},
that is, a tree with one node consisting of $c$ itself.  Therefore
we set:
\begin{align*}
  &\upeta\, (\Slice M)\, (i \,, c) &&= \\
  &&& \hspace{-2cm} \nd c\, (\lambda\, p \to\, \upeta\, M\, (\Typ M\, c\, p)) \\
  &&& \hspace{-1.3cm} (\lambda\, p \to \lf\, (\Typ M\, c\, p))
\end{align*}
Note that this definition would not be type correct without the
assumption that $M$ is definitionally right unital.  A similar remark
applies to the rest of the definitions of the slice monad in this
section.  Indeed, it is exactly the problem of completing the
definition of the slice monad without any assumptions of truncation
which motives the introduction of our monadic universe in the first
place.

Let us now sketch the definition of the multiplication in the slice
monad.  As hypotheses, we are given a tree
$\sigma : \Cns (\Slice M) (i \,, c)$ for some $i : \Idx M$ and
$c : \Cns M\, i$, as well as a decoration
$$\phi : (p : \Pos (\Slice M)\, \sigma) \to \Cns (\Slice M)\, (\Typ
(\Slice M)\, \sigma\, p)$$  In view of the preceding discussion, this
means that $\phi$ assigns to every position of $\sigma$ a tree which
multiplies to the constructor which inhabits that position.  The
multiplication of $\Slice M$ may intuitively be described as
``substituting'' each of these trees into the node it decorates.

The definition of $\upmu (\Slice M)$ will require an auxillary
function $\upgamma$ with the following type:
\begin{align*}
  \upgamma &: (M : \MM)\, \{i : \Idx M\}\, (c : \Cns M\, i) \\
           &\to (\sigma : \Cns (\Slice M)\, (i \,, c)) \\
           &\to (\phi : (p : \Pos M\, c) \to \Cns M\, (\Typ M\, c\, p)) \\
           &\to (\psi : (p : \Pos M\, c) \to \Cns (\Slice M)\, (\Typ M\, c\, p \,, \phi\, p)) \\
           &\to \Cns (\Slice M) (i \,, \upmu\, M\, c\, \phi)
\end{align*}
The intuition of this function is that $\upgamma$ \emph{grafts} the
tree specified by $\psi$ onto the appropriate leaf of the tree
$\sigma$ (indeed, $\gamma$ may be seen as an incarnation of
multiplication in the \emph{free} monad generated by the underlying
polynomial of $M$).  This function simply operates by induction and
may be defined as follows:
\begin{align*}
  &\upgamma\, M\, (\lf\, i)\, \delta\, \epsilon\, &&= \epsilon\, (\etapos M\, i) \\
  &\upgamma\, M\, (\nd\, c\, \delta\, \epsilon)\, \phi\, \psi\, &&= \nd\, c\, \delta'\, \epsilon'
\end{align*}
where we define
\begin{align*}
  \phi'\, p\, q &= \phi (\mupos M\, c\, \delta\, p\, q) \\
  \psi'\, p\, q &= \psi (\mupos M\, c\, \delta\, p\, q) \\
  \delta'\, p &= \upmu\, M\, (\delta\, p)\, (\phi' p) \\
  \epsilon'\, p &= \upgamma\, M\, (\epsilon\, p)\, (\psi' p)
\end{align*}
With this function in hand, we may complete the definition
of the multiplication in the slice monad as
\begin{align*}
  &\upmu\, (\Slice M)\, (\lf\, i)\, \phi &&= \lf\, i \\
  &\upmu\, (\Slice M)\, (\nd\, c\, \delta\, \epsilon)\, \phi &&= \upgamma\, M\, w\, \delta\, \psi
\end{align*}
where we put
\begin{align*}
  w &= \phi\, (\inl \ttt) \\
  \phi'\, p \, q &= \phi\, (\inr (p \,, q)) \\
  \psi\, p &= \upmu\, (\Slice M)\, (\epsilon\, p)\, (\phi'\, p)
\end{align*}
This definition then says that substitution is trivial on leaves, and
when we are looking at a node, we first retrieve the tree living at
this position (called $w$ above), and then graft to it the result of
recursively substituting in the remaining branches.

We refer the reader to the formalization for details on the remaining
constructions handling positions in the slice monad.\\

\subsection{Dependent monads}

Since the notion of dependent type is one of the primitive aspects of
Martin-L\"{o}f type theory, it is perhaps not surprising that we
quickly find ourselves in need of a dependent version of our
polynomial monads.  We note there is a potential point of confusion
here: while a dependent type can be thought of as a family of types
dependent on a base type, a dependent monad in our sense is \emph{not}
a family of monads.  Rather, it is a monad structure on dependent
families of indices and constructors indexed over the indices and
constructors of the base monad $M$.  Put another way, under the
equivalence between dependent types with \emph{domain} $A$ and
functions with \emph{codomain} $A$, our dependent monads over a base
monad $M$ correspond to monads $M'$ equipped with a \emph{cartesian
  homomorphism} to $M$. \footnote{In fact, it is entirely possible to
  add a monadic form of dependent sum to the list of monad
  constructors of the universe $\MM$ so that this statement becomes
  literally true.  As we will not have need of this construction in
  the present article, however, we omit the details.}  The advantage
of working in a dependent style, however, is that we do not need to
axiomatize the notion of homomorphism using propositional equalities
as it is encoded directly in the typing of the multiplication
operator.

To begin, let us postulate, for each monad $M : \MM$, a universe
$\MMd\, M : \UU$ of \emph{monads over $M$}.  
\[
  \MM\da : \MM \rightarrow \UU
\]
That is, for $M : \MM$, the inhabitants of $\MMd M$ are codes for
monads equipped with a cartesian morphism to $M$.  For this reason,
when we are given a monad $M$ and a dependent monad $M\da : \MMd\, M$,
we often speak of the pair $(M \,, M\da)$ as a \emph{monad extension}.

The decoding functions for dependent monads follow the same setup as
the non-dependent case, simply repeating each of the definitions
fiberwise.  And since the dependent case resembles so closely the
non-dependent one, we have attempted to systematically name dependent
versions of the the monadic structure introduced above by
appending a ``$\da$'' to the previously given name.  For example,
$\Idxd$ for the dependent version of the family $\Idx$ of indices.

As a first step, a dependent monad will be equipped with a polynomial
lying over the base polynomial.  This corresponds to the following
three dependent families:
\begin{align*}
  \Idxd &: \{M : \MM\} \to \MMd\ M \to \Idx\ M \to \UU\\
  \Cnsd &: \{M : \MM\}\ (M\da : \MMd\ M)\ \{i : \Idx\ M\}\\ 
        &\to \Idxd\ M\da\ i \to \Cns\ M\ i \to \UU\\
  \Typd &: \{M : \MM\}\ (M\da : \MMd\ M)\\
        &\to \{i : \Idx\ M\}\ \{i\da : \Idxd\ M\da\ i\}\\
        &\to \{c : \Cns\ M\ i\}\ (c\da : \Cnsd\ M\da\ i\da\ c)\\
        &\to \Pos\ M\ c \to \Idxd\ M\da\ (\Typ\ M\ c\ p)
\end{align*}
The reader will notice, however, that there is no analog of dependent
positions.  This is because we are modelling \emph{cartesian}
morphisms of monads, and therefore positions of a dependent
constructor $c\da : \Cnsd\ M\da\ i\da\ c$ should be the same as those
of the underlying constructor $c$.  By working fiberwise, we can
reflect this requirement directly in the type signature.

The monadic structure of a dependent monad simply operates fiberwise,
following the multiplication in the base monad:
\begin{align*}
  \upetad &: \{M : \MM\}\ (M\da : \MMd\ M)\\ 
            &\to \{i : \Idx\ M\} \to \Idxd\ M\da\ i\\ 
            &\to \Cnsd\ M\da\ i\, (\upeta\, M\, i)\\
  \upmud &: \{M : \MM\}\ (M\da : \MMd\ M)\\ 
           &\to \{i : \Idx\ M\}\ \{c : \Cns\ M\ i\} \\
           &\to \{\delta : (p : \Pos\ M\ c) \to \Cns\ M\ (\Typ\ M\ c\ p)\}\\
           &\to (i\da : \Idxd\ M\da\ i)\ (c\da : \Cnsd\ M\da\ i\da\ c)\\
          &\to (\delta\da : (p : \Pos\ M\ c) \to \\
          &\hspace{1cm} \Cnsd\ M\da\ (\Typd\ M\da\ c\da\ p)\ (\delta\ p))\\
           &\to \Cnsd\ M\da\ i\da\ (\upmu\ M\ c\ \delta)
\end{align*}
The fact that we require the multiplication of a family of dependent
constructors to live over the multiplication of the base constructors
(and similarly for the unit) is what guarantees the homomorphism
property alluded to above.

Our dependent monads must also be equipped with equational laws making
them compatible with the corresponding laws of the monads they live
over.  For example, the typing functions for $\upetad$ and $\upmud$
respect the indices of parameters, just as in the base case:
\begin{align*}
  \Typd\ M\da\ &(\upetad\ M\da\ i\da)\ p \leadsto i\da\\
  \Typd\ M\da\ &(\upmud\ M\da\ c\da\ \delta\da)\ p \leadsto \\
               &\Typd\ M\ (\delta\da\ (\muposfst\da M\da\ p))\ \\
               &\hspace{1.5cm} (\mupossnd\da M\da\ p)
\end{align*}
There are similar laws asserting the definitional associativity and
unitality of the multiplication, but as these all follow exactly
the same pattern, we omit the details here and refer the curious
reader to the implementation.

We remark that, because their positions are the same, decorations of
the dependent constructor $\upeta\da$ are essentially constant just as
in the case of $\upeta$, and there is therefore an analogous function
$\etadecd$ generating such decorations from a single piece of data
with a similar definition.  This function occurs occasionally in the
code below.

\subsection{Populating the dependent universe}

We now quickly describe dependent counterparts of the base monads and
monad constructors of the previous section.  As most of the
definitions are routine and easily deduced from the absolute case, the
presentation here is brief.

\subsubsection{The identity monad}

The dependent identity monad is parametrized by a type $A : \UU$ and
indexed over the identity monad $\Id$.  That is, we have a dependent
monad constructor of the form
\[
  \Idd : \UU \to \MMd\ \Id 
\]
Its polynomial part is defined by
\begin{align*}
  &\Idxd (\Idd A)\, \ttt &= A\\
  &\Cnsd (\Idd A)\, x\, \ttt &= \Unit\\
  &\Posd (\Idd A)\, \ttt\, \ttt &= \Unit\\
  &\Typd (\Idd A)\, \{i\da = x\}\ \ttt\ \ttt &= x
\end{align*}
As in the base case, the multiplication and unit all take values in
the unit type, making the structure essentially trivial.

\subsubsection{The dependent pullback monad}

Just as we can refine the indices of a base monad, so the dependent
pullback monad allows us to refine the indices of a dependent monad.
Its constructor takes the form
\begin{align*}
  \Pbd &: \{M : \MM\}\ (M\da : \MMd\ M)\ \{X : \Idx\ M \to \UU\}\\
       &\to (X\da : \{i : \Idx\ M\} \to \Idxd\ M\da\ i \to X\ i \to \UU)\\
       &\to \MMd\ (\Pb\ M\ X)
\end{align*}
Note that the family $X\da$ may also depend on elements of the
refining family $X$ for the base monad.  The underlying polynomial of
the dependent pullback is then defined as follows:
\begingroup
\addtolength{\jot}{1em}
\begin{align*}
  &\Idxd\ (\Pbd M\da\ X\da)\ (i, x) = \sum_{(i\da : \Idxd\, M\da\, i)} X\da\ i\da\ x\\
  &\Cnsd\ (\Pbd M\da\ X\da)\ (i\da , x\da)\ (c , \nu) =\\
  &\hspace{.5cm} \sum_{(c\da : \Cnsd\, M\da\, i\da\, c)} \prod_{(p : \Pos M\, c)} X\da\, (\Typd M\da\, c\da\, p)\, (\nu\, p) \\
  &\Typd\ (\Pbd M\da\, X\da)\ (c\da , \nu\da)\ p = \Typd\ M\, c\da\, p, \nu\da\, p
\end{align*}
\endgroup with multiplicative structure following fiberwise the rules
for the base pullback $\Pb M\, X$.

\subsubsection{The dependent slice monad}

Finally, the dependent slice monad extends the Baez-Dolan
slice construction to the dependent case.  Its monad constructor
is typed as follows:
\[ \Sliced : \{M : \MM\}\, (M\da : \MMd\, M) \to \MMd\ (\Slice M) \]
As for the absolute case, the indices are given by the dependent
constructors.  That is, we set
\[ \Idxd (\Sliced M\da)\ (i , c) = \sum_{i\da : \Idxd M\da\, i} \Cnsd
  M\da\, i\da\, c \] Similarly, the type of constructors
$\Cnsd (\Sliced M\da)$ are trees lying over a tree in the base.  This
corresponds to the following (rather verbose) inductive type:
\begin{align*}
  &\texttt{data}\hspace{1ex} \mathsf{Tree\da} : \{i : \Idx (\Slice M)\} \to (i\da : \Idxd (\Sliced M\da) \\
  &\hspace{1cm} \to \mathsf{Tree}\, i \to \UU \hspace{1ex} \texttt{where} \\
  &\lfd : \{i : \Idx M\}\ (i\da : \Idxd M\da\, i)\\
  &\hspace{.2cm}\to \Cnsd (\Sliced M\da)\, (i\da, \upetad M\da\, i\da)\, (\lf i) \\
  &\ndd : \{i : \Idx M\}\ \{c : \Cns M\, i\}\\ 
  &\hspace{.2cm}\to \{\delta : (p : \Pos M\, c) \to \Cns M\, (\Typ M\, c\, p)\}\\
  &\hspace{.2cm}\to \{\epsilon : (p : \Pos M\, c) \\
  &\hspace{1.2cm}\to \Cns (\Slice M)\, (\Typ M\, c\, p, \delta\, p)\}\\ 
  &\hspace{.2cm}\to \{i\da : \Idxd M\da\, i\} \to (c\da : \Cnsd M\da\, i\da\, c)\\
  &\hspace{.2cm}\to (\delta\da : (p : \Pos M\, c) \to \Cnsd M\da\, (\Typd M\da\, c\da\, p))\\ 
  &\hspace{.2cm}\to (\epsilon\da : (p : \Pos M\, c) \\
  &\hspace{1.2cm}\to \Cnsd (\Sliced M\da)\, (\Typd M\da\, c\da\, p, \delta\da\, p))\\ 
  &\hspace{.2cm}\to \Cnsd (\Sliced M\da)\, (i\da \,, \upmud M\da\, c\da\, \delta\da)\, (\nd c\, \delta\, \epsilon)
\end{align*}
The rest of the description of the dependent slice follows exactly the
same pattern: duplicating the definitions and laws of the base case
routinely in each fiber.


\section{Opetopic Types}
\label{sec:opetopic-types}

In this section, we show how to use the universes introduced above in
order to implement Baez and Dolan's definition of \emph{opetopic type}
\cite{BD98}.  We go on to explain how to use this definition to capture
the notion of \emph{weak $M$-algebra}, and finish with some examples.

\begin{defn}
  An \textbf{opetopic type over a monad $M$} is defined coninductively
  as follow:
  \begin{align*}
    &\texttt{record}\hspace{1ex} \OpType (M : \MM) : \UU_1 \hspace{1ex} \texttt{where} \\
    &\hspace{.5cm}\Car : \Idx M \to \UU \\
    &\hspace{.5cm}\Rel : \OpType (\Slice (\Pb M\, \Car))
  \end{align*}
\end{defn}
We see from the definition that an opetopic type consists of
an infinite sequence of dependent families
\[ \Car X \,, \Car (\Rel X) \,, \Car (\Rel (\Rel X)) \,, \dots \]
whose domain is the set of indices of a monad whose definition
incorporates all the previous families in the sequence.  Given an
opetopic type $X : \OpType M$, we will often denote this sequence of
dependent types more succinctly as just $X_0, X_1, X_2, \dots$ since
the destructor notation quickly becomes quite heavy.  We will use a
similar convention for the series of monads $M = M_0, M_1, M_2 \dots$
generated by the definition.  That is, we have:
\begin{equation}
  \label{eq:unfold}
  \begin{aligned}
    &M_0 = M && X_0 = \Car X : \Idx M \to \UU \\
    &M_1 = \Slice (\Pb M_0\, X_0) && X_1 = \Car (\Rel X) : \Idx M_1 \to \UU \\
    &M_2 = \Slice (\Pb M_1\, X_1) && X_2 = \Car (\Rel (\Rel X)) : \\
    &&&\hspace{2cm} \Idx M_2 \to \UU \\
    &\hspace{1cm} \vdots && \hspace{1cm} \vdots
  \end{aligned}
\end{equation}

Before describing the connection between opetopic types and weak
$M$-algebras, let us give some examples of how to think of the
resulting dependent families as ``fillers'' for a collection of
``shapes'' generated by the monad $M$.  For concreteness, we will fix
$M = \Id$ in our examples.  Given $X : \OpType \Id$, we can define the
type of \emph{objects} of $X$ as simply
\begin{equation}
  \label{eq:obj-defn}
  \begin{aligned}
    &\mathsf{Obj} : \UU \\
    &\mathsf{Obj} = \Car X\, \ttt
  \end{aligned}
\end{equation}
Next, after a single slice, $X$ provides us with a type of
\emph{arrows} between any two objects which can be defined as follows:
\begin{equation}
  \label{eq:arrow-defn}
  \begin{aligned}
    &\mathsf{Arrow} : (x\, y : \mathsf{Obj}) \to \UU \\
    &\mathsf{Arrow}\, x\, y = \Car (\Rel X)\, \\
    &\hspace{2cm} ((\ttt \,, y) \,, (\ttt , \etadec \Id\, (\Car X)\, x))
  \end{aligned}
\end{equation}
Furthermore, for a \emph{loop} $f$ in $X$, that is, an arrow with the
same domain and codomain, $X$ includes a family whose elements can be
thought of as ``null-homotopies of $f$'', and which is defined by
\begin{align*}
  &\mathsf{Null} : (x : \mathsf{Obj})\, (f : \mathsf{Arrow}\, x\, x) \to \UU \\
  &\mathsf{Null}\, x\, f = \Car (\Rel (\Rel X)) \\
  & \hspace{2cm} ((((\ttt \,, x) \,, (\ttt , \etadec \Id\, (\Car X)\, x)) \,, f) \,, \\
  & \hspace{2.5cm} \lf\, (\ttt \,, x) \,, \bot\mhyphen\mathsf{elim})
\end{align*}
More examples of shapes and filling families may be found in the
development.

\subsection{Weak Algebras and Fibrant Opetopic Types}
\label{sec:weak-alg}

We now wish to describe how an opetopic type $X : \OpType M$ encodes
the structure of a weak $M$-algebra.  Before we begin, it will be
convenient to adopt the following convention: recall that $X$ consists
of an infinite sequence of dependent types following the form of
Equation \ref{eq:unfold}.  In the discussion which follows, instead of
working with a fixed opetopic type $X$, we will rather just work with
abstract type families $X_0 , X_1 , \dots$ over monads
$M = M_0, M_1, \dots$ following the same pattern of dependencies. We
will then freely add new families of the form $X_i$ to our hypotheses
as they become necessary.  The advantage of working this way is that
our definitions are parameterized over just that portion of the
opetopic type which is necessary, as opposed to depending on the
entire opetopic type $X$ itself, and consequently, we will be able to
reuse our definitions and constructions starting at any point in the
infinite sequence generated by $X$.

We recall that for $M$ a polynomial monad, an \emph{$M$-algebra}
consists of a \emph{carrier family} $C : \Idx M \to \UU$ together with
a map
\[ \alpha : \{i : \Idx M\} \to [ M ]\, C\, i \to C\, i \] which
satisfies some equations expressing the compatibility of $\alpha$ with
the multiplication of $M$.  Indeed, it is the need for a complete
description of these equations in all dimensions which motivates the
present work.  Now, clearly the first dependent type
$X_0 : \Idx M \to \UU$ may serve as a carrier for an $M$-algebra
structure.  Let us now see what else this sequence of families
provides us with.

After one iteration, we obtain a type family $X_1 : \Idx M_1 \to \UU$,
and unfolding the definition of the indices of the slice and pullback
monads, we find that the domain of $X_1$ takes the form
\begin{align*}
  \sum_{(i : \Idx M)} \sum_{(x : X_0 i)}
  \sum_{(c : \Cns M\ i)} (p : \Pos M\, c) \to X_0\, (\Typ M\, c\, p)
\end{align*}
The elements of this type are 4-tuples $(i \,, x \,, c \,, \nu)$, and
we now observe that the three elements $i$, $c$ and $\nu$ are typed
such that they are exactly the arguments of the hypothetical algebra
map $\alpha$ introduced above.  We may regard the family $X_1$,
therefore, as a relation between triples $(i \,, c \,, \nu)$ and
elements $x : X_0\, i$, and in order to define a map $\alpha$, we only
need to impose that this relation is functional in the sense that
there is a \emph{unique} $x$ determined by any such triple.  When this
is the case, we will say that the family $X_1$ is
\emph{multiplicative}.  That is, we define:
\begin{align*}
  &\ismult : \{X_0 : \Idx M_0 \to \UU\}\, (X_1 : \Idx M_1 \to \UU) \to \UU \\ 
  &\ismult \{X_0\}\, X_1 = \{i : \Idx M\}\, (c : \Cns M\, i) \\
  &\hspace{1cm} \to (\nu : (p : \Pos M\, c) \to X_0\, (\Typ M\, c\, p))\\
  &\hspace{1cm} \to \iscontr (\sum_{x : X_0\, i} X_1\, (i \,, x \,, c \,, \nu))
\end{align*}
Supposing we are given a proof $m_1 : \ismult X_1$, we can define an
algebra map $\alpha$ as above by
\[\alpha\, (c \,, \nu) = \fst (\ctr (m_1\, c\, \nu)) \]
Furthermore, we will write
\[\alpha\wit\, (c \,, \nu) = \snd (\ctr (m_1\, c\, \nu)) \]
for the associated element of the relation
$X_1 (i \,, \alpha\, (c \,, \nu) \,, c \,, \nu)$ which witnesses
this multiplication.

Let us now suppose that our sequence extends one step further, that
is, that we are given a type family $ X_2 : \Idx M_1 \to \UU $ and a
proof $m_2 : \ismult X_2$.  We now show how to use this further
structure to derive some of the expected \emph{laws} for the algebra
map $\alpha$ we have just defined.  As a first example, we expect
$\alpha$ to satisfy a unit law: decorating a unit constructor with
some element $x$ and then applying $\alpha$ should return the element
$x$ itself.  In other words, we expect to be able to prove
\begin{align*}
  &\alpha^{\eta}\coh : \{i : \Idx M\} (x : X_0\, i) \\
  &\hspace{.3cm} \to \alpha\, (\upeta\, M\, i \,, \etadec M\, X_0\, x)\, \equiv x
\end{align*}
To prove this equality, let us define the following function:
\begin{align*}
  &\upeta\alg_{m_2} : \{i : \Idx M\} (x : X_0\, i) \\
  &\hspace{.3cm} \to X_1\, ((i \,, x) \,, (\upeta\, M\, i , \etadec M\, X_0\, x)) \\
  &\upeta\alg_{m_2} = \fst (\ctr (m_2\, (\lf (i \,, x))\, \botelim)) 
\end{align*}
Now we simply notice that the pairs
\[ \ctr (m_1\, (\upeta\, M\, i)\, (\etadec M\, X_0\, x)) \equiv (x \,,
  \upeta\alg\, x) \] must be equal as indicated, since they inhabit a
contractible space.  Projecting on the first factor gives exactly the
desired equation.

We also expect our algebra map $\alpha$ to satisfy an equation
expressing its compatibility with multiplication of the following
form:
\begin{align*}
  &\alpha^{\upmu}\coh : \{i' : \Idx M\}\, (c' : \Cns M\, i)\\
  &\hspace{.2cm} \to (\delta' : (p : \Pos M\, c) \to \Cns M (\Typ M\, c\, p))\\
  &\hspace{.2cm} \to (\nu' : (p : \Pos M\, c') (q : \Pos M\, (\delta'\, p)) \\
  &\hspace{3cm} \to X_0\, (\Typ M\ (\delta'\, p)\, q)) \\
  &\hspace{.2cm} \to \alpha\, (\upmu\, M\, c' \, \delta')\, (\lambda p \to \nu'\, (\muposfst p)\, (\mupossnd p)) \equiv\\
  &\hspace{.7cm} \alpha\, c'\, (\lambda p \to \alpha\, (\delta'\, p)\, (\nu'\, p))
\end{align*}
We note that this equation is simply the type theoretic translation of
the familiar commutative diagram
\[
  \begin{tikzcd}
    {[ M ]\, [ M ]\, X_0} \ar[r,"\upmu_{X_0}"] \ar[d,"{[ M ]\, \alpha}"'] &
    {[ M ]\, X_0} \ar[d,"\alpha"] \\
    {[ M ]\, X_0} \ar[r,"\alpha"'] & X_0
  \end{tikzcd}
\]
To prove this axiom, we use $m_2$ to define the following
multiplication operation on elements of the family $X_1$:
\begin{align*}
  &\upmu\alg_{m_2} : \{i : \Idx M\}\, (c : \Cns M\, i)\\
  &\hspace{.3cm}\to (\nu : (p : \Pos M c) \to X_0\, (\Typ M\, c\, p)) \\
  &\hspace{.3cm}\to (\delta : (p : \Pos M c)\\
  &\hspace{2cm} \to  \Cns (\Pb M\, X_0) (\Typ (\Pb M\, X_0)\, (c, \nu)\, p)) \\
  &\hspace{.3cm}\to (x_0 : X_0\, i)\, (x_1 : X_1\, (i \,, x_0 \,, c \,, \nu)) \\
  &\hspace{.3cm}\to (\bar{x} : (p : \Pos M\, c) \to X_1 (\Typ (\Pb M\, X_0) (c \,, \nu) \,, \delta\, p)) \\
  &\hspace{.3cm}\to X_1 (i \,, x_0 \,, \upmu\, (\Pb M\, X_0)\, (c \,, \nu)\, \delta) \\
  &\upmu\alg_{m_2} = \fst (\ctr (m_2\, \sigma\, \theta))
\end{align*}
where
\begin{align*}
  &\sigma = \nd (c \,, \nu)\, \delta\, (\lambda p \to \upeta\, M_1\, ((\Typ M\, c\, p \,, \nu\, p) \,, \delta\, p)) \\
\end{align*}
is the two-level tree consisting of a base node $(c , \nu)$, as well
as a second level of constructors specified by the decoration
$\delta$, and $\theta$ is the decoration of the nodes
of $\sigma$ by elements of $X_1$ defined by:
\begin{align*}
  &\theta\, (\inl \ttt) = x_1 \\
  &\theta\, (\inr (p \,, \inl \ttt)) = \bar{x}\, p \\
\end{align*}
Now instantiating our function $\upmu\alg_{m_2}$ with arguments
\begin{align*}
  &c = c' &&x_1 = \alpha\mhyphen\mathsf{wit}\, (c \,, \nu) \\
  &\nu\, p = \alpha\, (\delta'\, p \,, \nu'\, p) &&\delta\, p = (\delta'\, p \,, \nu' \, p) \\
  &x_0 = \alpha\, (c \,, \nu) &&\bar{x}\, p = \alpha\mhyphen\mathsf{wit}\, (\delta'\, p \,, \nu'\, p) 
\end{align*}
we find that the pairs
\begin{align*}
  &\ctr (m_1\, (\upmu\, M\, c'\, \delta')\, (\lambda p \to \nu'\, (\muposfst p)\, (\mupossnd p))) \equiv \\
  &\hspace{.3cm} (\alpha\, (c \,, \nu) \,,  \upmu\alg_{m_2}\, c\, \nu\, x_0\, x_1\, \delta\, \bar{x})
\end{align*}
again inhabit a contractible space, whereby their first components are
equal, giving the desired equation.

We may think of the functions $\upeta\alg_{m_2}$ and $\upmu\alg_{m_2}$
defined above as the nullary and binary cases of a multiplicative
operation on the \emph{relations} of our algebra structure.  The key
insight, as we have seen, is that this multiplicative structure
encodes exactly the \emph{laws} for the algebra map $\alpha$ defined
one level lower.  Similarly, if we are able to extend our sequence on
\emph{further} step to a family $X_3$ which is itself multiplicative,
then we will be able to show that the operations $\upeta\alg_{m_2}$
and $\upmu\alg_{m_2}$ \emph{themselves satisfy unit and associativity
  laws}, and this in turn encodes the ``2-associativity'' and
``2-unitality'' of the algebra map $\alpha$.  This motivates
the following definition:

\begin{defn}
  An opetopic type $X$ over a monad $M$ is said to be \textbf{fibrant}
  if we are given an element of the following coinductively defined
  property:
  \begin{align*}
    &\texttt{record}\hspace{1ex} \isfibrant \{M : \MM\}\, (X : \OpType M) : \UU\\
    &\hspace{.3cm} \texttt{where} \\
    &\hspace{.5cm} \carmult : \ismult M\, (\Car (\Rel X)) \\
    &\hspace{.5cm} \relfib : \isfibrant (\Rel X)
  \end{align*}
\end{defn}
Fibrant opetopic types, therefore, are our definition of infinitely
coherent $M$-algebras, with the multiplicativity of the relations
further in the sequence witnessing the higher dimensional laws
satisfied by the structure earlier in the sequence.

\subsection{Higher structures}
\label{sec:higher-structures}

We now use the preceding notions to define a number of coherent
algebraic structures.  A first example is that we obtain an internal
definition of the notion of $\infty$-groupoid as follows: 
\begin{defn}
  An \textbf{$\infty$-groupoid} is a fibrant opetopic type over
  the identity monad.  That is,
\[ \ooGrp = \dsum{X : \OpType \Id} \isfibrant X \] 
\end{defn}
\noindent We will attempt to justify the correctness of this
definition in the sections which follow.

Next, it happens that the monad $\Slice \Id$ is in fact the monad
whose algebras are monoids, and consequently, our setup leads
naturally to the definition of an $\mathbb{A}_\infty$-type, that is, a
type with a coherently associative binary operation.

\begin{defn}
  An \textbf{$\mathbb{A}_\infty$-type} is a fibrant opetopic type
  over the first slice of the identity monad.
  \[ \mathbb{A}_\infty\mhyphen\mathsf{type} = \dsum{X : \OpType
      (\Slice \Id)} \isfibrant X \]
\end{defn}
Furthermore, the notion of $\mathbb{A}_\infty$-group can now be
defined by imposing an invertibility axiom.  A classical theorem of
homotopy theory asserts that the type of $\mathbb{A}_\infty$-groups is
equivalent to the type of pointed, connected spaces via the loop-space
construction.  It would be interesting to see if the techniques of
this article lead to a proof of this fact in type theory.

The notion of $\infty$-category can also be defined using this setup.
Recall that an opetopic type over the identity monad $\Id$ has both a
type of objects and a type of arrows (Equations \ref{eq:obj-defn} and
\ref{eq:arrow-defn}).  In the definition of $\infty$-groupoid above,
the invertibility of the arrows in the underlying opetopic type is a
consequence of the fact that the family of arrows is assumed to be
multiplicative.  Consequently, we obtain a reasonable notion of a
\emph{pre-$\infty$-category} by simply dropping this assumption, and
only requiring fibrancy after one application of the destructor
$\Rel$:
\[ \preooCat = \dsum{X : \OpType \Id} \isfibrant (\Rel X) \] The
prefix ``pre'' here refers to the fact that this definition is missing
a completeness axiom asserting that the invertible arrows coincide
with paths in the space of objects, that is, an axiom of
\emph{univalence} in the sense of \cite{ahrens2015univalent}.  Such an
axiom is easily worked out in the present setting, but as it would
distract us slightly from the main objective of the present work, we
will not pursue the matter here.

\section{The $\infty$-groupoid associated to a type}
\label{sec:infty-group-assoc}

In this section, we use the machinery we have set up to produce an
$\infty$-groupoid associated to any type and eventually prove it is
unique.  As a first step, we will need a source of opetopic types.
Here is where the notion of dependent monad becomes important: we now
show that every dependent monad gives rise to an opetopic type.  The
reason for this phenomenon is simple: since our dependent monad
constructors mirror the monad constructors of the absolute case, any
monad extension $(M \,, M\da)$ in fact gives rise to a \emph{new}
monad extension as follows:
\begin{align*}
  M &\quad \mapsto \quad \Slice (\Pb M\, (\Idxd M\da)) \\
  M\da &\quad \mapsto \quad \Sliced (\Pbd M\da\, (\lambda\, j\, k \to j \equiv k))
\end{align*}
Notice how by pulling back along $\Idxd M\da$, the identity type gives
us a canonical family along which to apply the $\Pbd$ constructor.
Iterating this construction, then, we find that associated to every
monad extension $(M \,, M\da)$, is an infinite sequence
\[ (M \,, M\da) = (M_0 \,, M\da_0), (M_1 \,, M\da_1), (M_2 \,, M\da_2),
  \dots \] where $(M_{i+1} \,, M\da_{i+1})$ is obtained from
$(M_i \,, M\da_i)$ by the above transformation.

The above construction provides us with our desired source of opetopic
types.  Formally, we define (using copattern notation)
\begin{align*}
  &\OvrOpType M\, M\da : \OpType M  \\
  &\Car (\OvrOpType M\, M\da) = \Idxd\, M\da \\
  &\Rel (\OvrOpType M\, M\da) = \\
  &\hspace{.5cm}\OvrOpType (\Slice (\Pb (\Idxd\, M\da))) \\
  &\hspace{.8cm} (\Sliced (\Pbd\, M\da\, (\lambda\, j\, k \to j \equiv k)))
\end{align*}
Specializing to the case of the identity monad, we obtain the
following:

\begin{defn}
  For a type $A : \UU$, the \textbf{underlying opetopic type
    of $A$} is defined to be the opetopic type associated
  to the dependent identity monad determined by $A$.  That is,
  the opetopic type
  \[ \OvrOpType \Id\, (\Id\da\, A) \]
  in the notation of the previous paragraph.
\end{defn}
In order to show that every type $A$ determines an $\infty$-groupoid
in our sense, our next task is to show that this opetopic type is
in fact fibrant.

\subsection{Algebraic Extensions}
\label{sec:algebraic-extensions}

Let $M : \MM$ and $M\da : \MMd$.  We will say that the extension
$(M, M\da)$ is \emph{algebraic} if we have a proof
\begin{align*}
  &\isalgebraic : (M : \MM)\, (M\da : \MMd)\, \to \UU \\
  &\isalgebraic = \{i : \Idx M\}\, (c : \Cns M\, i) \\
  &\hspace{.3cm} \to (\nu : (p : \Pos M\, c) \to \Idxd M\da\, (\Typ M\, c\, p)) \\
  &\hspace{.3cm} \to \iscontr \left (\sum_{(i\da : \Idxd M\da)} \sum_{(c\da : \Cnsd M\da\, i\da)} \Typd M\da\, c\da \equiv \nu \right )
\end{align*}
An algebraic extension should be thought of as roughly analogous to a
generalized kind of opfibration: if we think of the constructors as
generalized arrows between their input indices and output, then the
hypothesis says we know a family of lifts over the source of our
constructor, and the conclusion is that there exists a unique
``pushforward'' consisting of a lift over the output as well as a
constructor connecting the two whose typing function agrees with the
provided input lifts.  Such a hypothesis is one way of encoding an
$M$-algebra, which motivates the name for this property.
See~\cite[Section 6.3]{leinster2004higher}.

The main use of the notion of algebraic extension is the following
lemma, whose proof is entirely straightforward:
\begin{lem}
  \label{lem:alg-to-fib}
  Suppose the pair $(M \,, M\da)$ is an algebraic extension.  Then
  the relation $\Idxd M\da_1$ is multiplicative.
\end{lem}
Consequently, just as dependent monads are a source of opetopic types,
algebraic extensions can be thought of as a source of multiplicative
relations.  Hence if we want to prove fibrancy of the opetopic type
associated to a monad extension, we will need to know which of the
extensions in the generated sequence are algebraic.  Our main theorem
is that after a single iteration of the slice construction,
\emph{every} monad extension becomes algebraic.  That is

\begin{thm}
  \label{thm:slice-alg}
  Let $(M \,, M\da)$ be a monad extension.  Then slice extension
  $(M_1 \,, M\da_1)$ is algebraic.
\end{thm}

A proof can be found in the extended version of this article
\cite{allioux2021types}.  The importance of the theorem is that it has
the following immediate corollaries:

\begin{cor}
  \label{cor:alg-is-fib}
  Let $(M \,, M\da)$ be an algebraic extension.  Then the opetopic
  type $\OvrOpType M\, M\da$ is fibrant.
\end{cor}

\begin{proof}
  The base case of the coinduction is Lemma \ref{lem:alg-to-fib} and
  the coinductive case is covered by Theorem
  \ref{thm:slice-alg}.
\end{proof}

\begin{cor}
  There is a map $\Gamma : \UU \to \ooGrp$.
\end{cor}

\begin{proof}
  Let $A : \UU$ be a type.  A short calculation shows that
  the monad extension $(\Id \,, \Id\da\, A)$ is algebraic.
  The result therefore follows from Corollary \ref{cor:alg-is-fib}.
\end{proof}

\subsection{Uniqueness of the Groupoid Structure}

We now turn to the task of showing the map $\Gamma : \UU \to \ooGrp$
is an equivalence.  Observe that there is a forgetful map
$\Upsilon : \ooGrp \to \UU$ which is given by extracting the type of
objects (Equation \ref{eq:obj-defn}) from the opetopic type underlying
a groupoid $G : \ooGrp$.  It is readily checked that the composite
$\Upsilon \circ \Gamma$ is definitionally the identity, and so what
remains to be shown is that any $G : \ooGrp$ is equivalent to $\Gamma$
applied to its type of objects.

Unwinding the definitions, we find that we are faced with the following
problem: suppose we are given a monad extension $(M \,, M\da)$ as well
as a opetopic type $X : \OpType M$.  Under what hypotheses can we
prove that $X \simeq_{o} \OvrOpType M\, M\da$ (where $\simeq_o$
denotes an appropriate notion of equivalence of opetopic types)?  A
first remark is that the opetopic type $\OvrOpType M\, M\da$ is
completely determined by the algebraic structure of the dependent
monad $M\da$.  Therefore, at a minimum, we must assume that the data
of the opetopic type $X$ is equivalent to the data provided by $M\da$
wherever they ``overlap''.

To see what this means concretely, let us begin at the base of the
sequence, writing $Z = \OvrOpType M\, M\da$ to reduce clutter.  Now,
the family $Z_0 : \Idx M \to \UU$ is, by definition, given by the
family of dependent indices $\Idxd M\da$ of the dependent monad
$M\da$.  On the other hand, without additional hypotheses, the
opetopic type $X$ only provides us with some abstract type family
$X_0 : \Idx M \to \UU$.  Clearly, then, we will need to assume an
equivalence $e_0 : (i : \Idx M) \to \Idxd M\da\, i \simeq X_0\, i$ in
order to have any chance to end up with the desired equivalence of
opetopic types.

Moving on to the next stage, here we find that $Z_1$ is given
by the dependent indices
\[ \Idxd M\da_1 : \Idx M_1 \to \UU \] of the first iteration of the
dependent slice-pullback construction.  Unfolding the definition,
these are of the form
\begin{align*}
  (\Idxd M\da_1)\, (i \,, j \,, c \,, v)
  &= \\
  & \hspace{-2cm} \sum_{(j : \Idxd M\da\, i)} \sum_{(r : j \equiv j')} 
    \sum_{(d : \Cnsd M\da\, c)} (\Typd M\da\, d \equiv \nu)
\end{align*}
With the 4-tuple $(i \,, j \,, c \,, v)$ as in Equation
\ref{eq:slc-idx}.  We notice that much of the data here is redundant:
by eliminating the equality $r$ and the equality relating $\nu$ to the
typing function of $d$, we find that the dependent indices are
essentially just dependent constructors of $M\da$, slightly reindexed.
In other words, a dependent equivalence
\[ e_1 : (i : \Idx M_1) \to \Idxd M\da_1 \simeq_{e_0} X_1 \] over the
previous equivalence $e_0$ amounts to saying that the relations of the
family $X_1$ ``are'' just the dependent constructors of $M\da$ (again,
reindexed according to the typing of their input and output
positions).  As this is again part of the data already provided by the
dependent monad $M\da$, we will additionally need to add such an
equivalence to our list of hypotheses.

To recap: assuming the equivalences $e_0$ and $e_1$ amounts to
requiring that the first two stages of the opetopic type $X$ are
equivalent to the indices and constructors of the dependent monad
$M\da$, respectively.  What structure remains? Well, the dependent
constructors of $M\da$ are equipped with the unit and multiplication
operators $\upetad$ and $\upmud$.  But now, recall from Section
\ref{sec:weak-alg} that if the family of relations $X_1$ extends
further in the sequence to a family $X_2$ and we have a proof
$m_2 : \ismult X_2$, then the family $X_1$ can be equipped with a
multiplicative structure given by the functions $\upeta\alg_{m_2}$ and
$\upmu\alg_{m_2}$ defined there.  This is the case in the current
situation, if we assume that the opetopic type $X$ is fibrant (in
fact, we only need assume that $\Rel X$ is fibrant to make this
statement hold).  Therefore, the last piece of information in order
that $X$ ``completely agrees'' with the dependent monad $M\da$ is that
the equivalence $e_1$ is additionally a \emph{homomorphism}, sending
$\upetad$ to $\upeta\alg_{m_2}$ and $\upmud$ to $\upmu\alg_{m_2}$. Our
theorem now is that this data suffices to prove an equivalence of
opetopic types:

\begin{thm}
  \label{thm:slice-unique}
  Suppose $(M , M\da)$ is a monad extension and $X : \OpType M$
  an opetopic type such that $\Rel X$ is fibrant.  Moreover,
  suppose we are given the data of
  \begin{itemize}
  \item An equivalence $e_0 : (i : \Idx M) \to \Idxd M\da\, i \simeq X_0\, i$
  \item An equivalence $e_1 : (i : \Idx M_1) \to \Idxd M\da_1 \simeq_{e_0} X_1$ over $e_0$
  \item Proofs that $s : \upetad M\da \equiv_{e_0 , e_1} \upeta\alg_{m_2}$ and $t : \upmud M\da \equiv_{e_0 , e_1} \upmu\alg_{m_2}$    
  \end{itemize}
  Then there is an equivalence of opetopic types
  \[ X \simeq_{o} \OvrOpType M\, M\da \]
\end{thm}

We have taken some liberties in the presentation of this theorem
(strictly speaking, we have not stated precisely in what sense the
second equivalence $e_1$ is ``over'' the equivalence $e_0$, nor
precisely what equality is implied by symbol the $\equiv_{e_0,e_1}$)
but these omissions can be made perfectly rigorous by standard
techniques, and we feel the statement above conveys the essential
ideas perhaps more clearly than a fully elaborated statement, which
would require a great deal more preparation, not to mention space.
See the appendix of the extended version of this article for a proof
\cite[{Theorem \ref{thm:slice-unique}}]{allioux2021types}.

We at last obtain our desired equivalence:

\begin{thm}
  \label{thm:types-are-oogrps}
  The map
  \[ \Gamma : \UU \to \ooGrp \]
  is an equivalence.
\end{thm}

\begin{proof}
  Given $G : \ooGrp$, we let $A : \UU$ be its type of objects.  We now
  apply Theorem~\ref{thm:slice-unique} with $M = \Id$ and
  $M\da = \Idd A$.  We may take $e_0$ to be the identity.  The
  equivalence $e_1$ is a consequence of \cite[Theorem 5.8.2]{hottbook}
  and the required equalities are a straightforward calculation.
\end{proof}


\section{Conclusion}
\label{sec:concl}

We have presented an approach to defining higher coherent structures
in homotopy type theory by equipping type theory with a primitive set
of structures collected into a universe $\MM$ of polynomial monads,
and demonstrated that this approach can be used to prove non-trivial
theorems about these structures.  In this brief final section, we
compare some related approaches and survey some of the possible
directions and applications.

\subsection{Future Directions}
\label{sec:future-directions}

\subsubsection{Symmetric Structures}

A natural class of structures which escapes the capabilities of our
current approach is that of \emph{symmetric structures}, that is,
those which incorporate higher analogs of commutativity.  Examples
would include $\mathbb{E}_{\infty}$ groups and monoids, symmetric
monoidal categories, and general $\infty$-operads and their algebras.

\subsubsection{Higher Category Theory}

As we have seen, one higher structure which \emph{is} amenable to
treatment by our methods is that of an $\infty$-category.  An obvious
point to follow up on, then, is how much of the well developed theory
of $\infty$-categories can be formalized in this manner.

\subsubsection{A General Theory of Structures}

As we have mentioned in the introduction, we see the present work as a
first step towards a general theory of types and structures.  And
though we feel certain that at least some of the ideas of the present
work will carry over to such a theory, a complete picture of the basic
principles remains to be understood.  Moreover, a careful
investigation of the interaction of our techniques with univalent
implementations of type theory (such as \emph{cubical} type theory)
also remains for future work.

Accompanying such a general theory, we anticipate a deeper
investigation of the meta-theoretic properties of our proposed
approach.  For example, the Agda implementation is limited by the
expressivity of rewrite rules, and complicated by the explicit
universe construction, while a proper extension of MLTT would allow
for the investigation of meta-theoretic properties like decidability
of type checking and strong normalization using techniques like
normalization-by-evaluation (and potentially settling the conjecture
of \ref{sec:mnd-struct}). Furthermore, we have not touched at all on
the potential models of our system, topic which deserves we feel
deserves careful attention.



\bibliographystyle{IEEEtran}
\bibliography{IEEEabrv,../refs}

\end{document}